\documentclass[journal]{IEEEtran}  
       \usepackage{blindtext, graphicx,booktabs}
       \usepackage{amsmath,amssymb}
       \usepackage{cite}
       \usepackage{breqn}
       
       \usepackage[utf8]{inputenc}
       \usepackage[english]{babel}
       \allowdisplaybreaks[4]
       \usepackage[utf8]{inputenc}
       \usepackage{amsthm}

\IEEEoverridecommandlockouts

\title{\LARGE \bf
The Impact of Road Configuration in V2V-based Cooperative Localization: Mathematical Analysis and Real-world Evaluation
}

\newtheorem{theorem}{Theorem}

\begin{document}

\author{Macheng Shen$^{1}$, Jing Sun$^{1}$ and Ding Zhao$^{2}$
\thanks{*This work is funded by the Mobility Transformation Center at the University of Michigan with grant No. N021548.}
\thanks{$^{1}$M. Shen and J. Sun are with the Department
of Naval Architecture and Marine Engineering, University of Michigan, Ann Arbor,
MI, 48109. (e-mail: macshen@umich.edu; jingsun@umich.edu)}%
\thanks{$^{2}$D. Zhao (corresponding author, e-mail: zhaoding@umich.edu) is with the University of Michigan Transportation Research Institute, Ann Arbor, MI, 48109. }%
}

\maketitle
\thispagestyle{empty}
\pagestyle{empty}
\allowdisplaybreaks

\begin{abstract}

Cooperative map matching (CMM) uses the Global Navigation Satellite System (GNSS) position information of a group of vehicles to improve the standalone localization accuracy. It has been shown, in our previous work, that the GNSS error can be reduced from several meters to sub-meter level by matching the biased GNSS positioning to a digital map with road constraints. While further error reduction is expected by increasing the number of participating vehicles, fundamental questions on how the vehicle membership within CMM affects the performance of the CMM results need to be addressed to provide guidelines for design and optimization of the vehicle network. This work presents a theoretical study that establishes a framework for quantitative evaluation of the impact of the road constraints on the CMM accuracy. More specifically, a closed-form expression of the CMM error in terms of the road constraints and GNSS error is derived based on a simple CMM rule. The asymptotic decay of the CMM error as the number of vehicles increases is established and justified through numerical simulations. Moreover, it is proved that the CMM error can be minimized if the directions of the roads on which the connected vehicles travel obey a uniform distribution. Finally, the localization accuracy of CMM is evaluated based on the Safety Pilot Model Deployment and Pillar dataset of Ann Arbor traffic flow collected over three years period. The contributions of this work include establishing a theoretical foundation for CMM as well as providing insight and motivation for applications of CMM.
\end{abstract}

\section{INTRODUCTION}
Low-cost Global Navigation Satellite Systems (GNSS) are used for most mobile applications, whose localization accuracy are typically in the range of several meters. The limited accuracy of the low-cost GNSS is due to the atmospheric error and satellite error, collectively referred to as the common error, as well as receiver noise and multipath error, referred to as the non-common error.\\
\indent Enhancement of standalone GNSS has been an active research area in the recent years. Various approaches have been developed, such as sensor fusion (\cite{ITS_Liu}, \cite{CloudTech}, \cite{Li}, \cite{ACC}), cooperative localization (\cite{ITS_Cruz}, \cite{IEEE_magazine}, \cite{IV}) and map matching (\cite{IAIN}, \cite{rohani2016novel}, \cite{shen2016enhancement}, \cite{ITS_Knoop}). For instance, Liu et al. \cite{ITS_Liu} proposed a novel robust cubature Kalman filter to improve the performance of data fusion with uncertain sensor observations under a GNSS/Dedicated Short Range Communications (DSRC) integrated architecture for cooperative localization. Cruz et al. \cite{ITS_Cruz} employed a two-stage Bayesian filter to track the vehicle's position that fuses information from smartphone inertial sensors, V2V signal strength measurements, velocity, Global Positioning System (GPS) position and map information. Knoop et al. \cite{ITS_Knoop} conducted a test for lane determination with GPS Precise Point Positioning which exploited the fact that most driving is within a lane so as to achieve a decimeter accuracy. \\
\indent These works present promising improvements on GNSS localization, while most of them utilize additional information from different sources. Improving the localization accuracy of these widespread GNSS without incurring additional hardware and infrastructure costs has motivated recent research activities on Cooperative Map Matching (CMM) \cite{rohani2016novel}, \cite{shen2017improving}, \cite{shen2016enhancement}. CMM has been shown capable of improving Global Navigation Satellite System (GNSS) positioning of a group of connected vehicles through estimation and correction of the GNSS common localization error. The estimation error for the common localization error, which determines the accuracy of the CMM localization, is hereafter referred to as the CMM error. 

\begin{figure}[htbp]
  \centering
  \includegraphics[width=0.95\columnwidth]{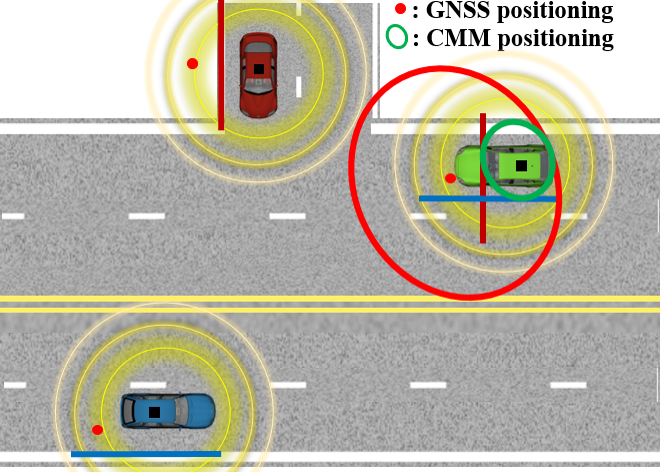}   
  \caption{Illustration of the positioning correction through CMM}
  \label{fig. CMM}
\end{figure}
%
%
\indent Fig. \ref{fig. CMM} illustrates how error-corrupted GNSS information can be corrected through a naive CMM approach. The GNSS solutions, denoted as red dots, of all the three vehicles are biased away from their true positions, while the biases are highly correlated. Virtual road constraints from the red and the blue vehicles are applied to restrict the positioning of the green vehicle. The underlying assumption for the application of the virtual constraints is that the GNSS positioning biases of the three vehicles are the same. This example manifests the impact of road constraints on the CMM localization. \\   
\indent It is desirable to quantify this impact so that we can estimate the CMM accuracy as well as optimize the vehicular network to achieve maximal error reduction. In our previous conference paper \cite{shen2017impact}, the correlation between the CMM error and the richness of the road constraints is quantified analytically. More specifically, a closed form expression of the CMM error in terms of the road angles as well as the GNSS error is derived based on a simple CMM rule that neglects the probabilistic property of the GNSS error.\\  
\indent In this work, we extend our previous work with a theoretical analysis on the optimal distribution of road directions associated with the vehicles that participate in the CMM and evaluation of the CMM accuracy to data collected in Ann Arbor. These results constitute the theoretical foundation of CMM as well as provide a guideline for the implementation of CMM.\\
\indent The organization of this paper is described as follows. In Section 2, an analytic expression of CMM as a function of the road configuration and the non-common error is derived. In Section 3, asymptotic formulas of the expectation of the CMM error with respect to Gaussian distributed non-common error are derived for uniformly or Bernoulli distributed road directions. In Section 4, simulation results are presented to demonstrate and justify the applicability of the theoretical results. In Section 5, it is proved that the uniform distribution of road angles associated with participating vehicles is optimal for minimizing the mean square CMM error. In Section 6, the localization accuracy of CMM is evaluated based on the Safety Pilot Model Deployment and Pillar 1 (SPMDP1) dataset collected in Ann Arbor over three years period. In Section 7, the contributions and conclusions are summarized.
\section{Derivation of the CMM error }
In this section, we propose a framework of vehicle positioning within a reference road framework to facilitate the analytic investigation.\\
\indent The following assumptions are essential in our exposition:
\begin{enumerate}
 \item The GNSS common error is the same for all the connected vehicles within the vehicular network. 
 \item The road side can be locally approximated as a straight line.
 \item The GNSS non-common error is random with Gaussian distribution.
 \item The GNSS non-common error is small enough such that the exact expression for the CMM error can be approximated by its first order linearization with respect to the non-common error.\\
 \end{enumerate}
 \indent The first assumption is often satisfied as long as the connected vehicles are geographically close to each other, for example, within several miles. The second and the fourth assumptions are made for mathematical convenience. If they are violated, however, the exact expression of the CMM error will still be valid but the asymptotic approximation will be inaccurate. The third assumption has been experimentally verified in \cite{shen2017improving} under open sky conditions.
 
\begin{figure}[htbp]
  \centering
  \includegraphics[width=0.8\columnwidth]{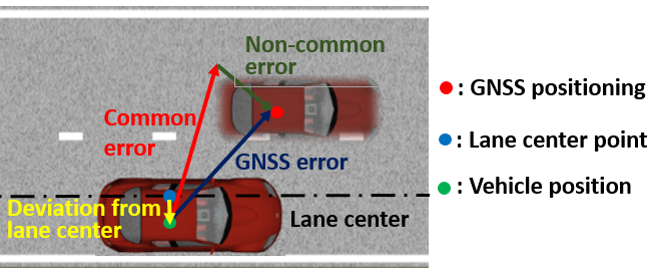}   
  \caption{Illustration of the notations used in CMM}
  \label{notation}
\end{figure}
\indent Consider a network of connected vehicles consisting of $N$ vehicles. The coordinate of GNSS positioning of the $i$-th vehicle, $x^G_i$, can be decomposed into a superposition of four elements: the coordinate of a point on the corresponding lane center, $x^L_i$, the deviation of the vehicle from the lane center, $X^D_i$, the GNSS common error, $x^C_i$, and the GNSS non-common error, $x^N_i$, as mathematically expressed in (\ref{eq1}) and graphically illustrated in Fig. \ref{notation}. The grayscale image of the vehicle represents the GNSS positioning.
\begin{equation}
x^G_i=x^L_i+x^D_i+x^C+x^N_i, i=1,2,...,N,
\label{eq1}
\end{equation}
\indent The fact that all the vehicles travel on the roads can be expressed as a set of inequalities
\begin{equation}
g_i(x^G_i-x^C-x^N_i)<0.
\label{eq2}
\end{equation}
\indent Applying the second assumption, the constraint functions $g_i$ have simple analytic forms
\begin{equation}
g_i(x)=(x-x^L_i)\cdot n_i-w,
\end{equation}
where $\{\cdot\}$ is the dot product operator, $n_i$ is the unit vector normal to the lane center point towards outside of the road and $w$ is the half width of the lane. \\
\indent The intersection of all the inequalities in (\ref{eq2}) can be interpreted as the feasible set of the common error given the GNSS positioning and the non-common error. The non-common error is, however, unknown to the implementation of CMM. Thus, the following approximation of the feasible set by neglecting the non-common error is used instead of the exact feasible set,
\begin{equation}
\begin{aligned}
\Omega &=\{\tau|\bigcap\limits_{i=1}^{N}g_i(x^G_i-\tau)<0\}\\
&=\{\tau|\bigcap\limits_{i=1}^{N}g_i(x^L_i+x^C+\tilde{x}^N_i-\tau)<0\}\\
&=\{\tau|\bigcap\limits_{i=1}^{N}\tilde{g}_i(x^C+\tilde{x}^N_i-\tau)<0\},
\end{aligned}
\label{eq4}
\end{equation}
where 
\begin{equation}
\tilde{x}^N_i \triangleq x^D_i+x^N_i
\end{equation}
and 
\begin{equation}
\tilde{g}_i(x)\triangleq g_i(x+x^L_i)=x\cdot n_i-w.
\end{equation}
Eq. (\ref{eq1}) has been applied in (\ref{eq4}) to derive the second equality from the first one. A point estimator of the common error is taken as the average over the approximate feasible set $\Omega$,
\begin{equation}
\hat{x}^C=\frac{1}{S}\int\limits_\Omega \tau dA, S=\int\limits_\Omega dA,
\end{equation}
where $\tau$ is the dummy integration variable and $dA$ is the area element.\\
\indent The estimation error, that is the difference between the true common error and the estimated common error (referred to as CMM error), is of practical interest, which can be evaluated,
\begin{equation}
\begin{aligned}
e&=x^C-\hat{x}^C=x^C-\frac{1}{S}\int\limits_\Omega \tau dA\\
&=\frac{1}{S}\int\limits_\Omega (x^C-\tau) dA=\frac{1}{S}\int\limits_{\Omega'} \tau' dA,
\end{aligned}
\label{eq8}
\end{equation}
where 
\begin{equation}
\tau'=x^C-\tau,
\end{equation}
and
\begin{equation}
\Omega'=\{\tau'|\bigcap\limits_{i=1}^{N}\tilde{g}_i(\tilde{x}^N_i+\tau')<0\}.
\label{eq10}
\end{equation}
\indent Eq. (\ref{eq8}) and (\ref{eq10}) states that the CMM error equals to the geometric center of the intersection of the road constraints perturbed by the composite non-common error $\tilde{x}^N_i$. \\
\indent The expectation of the square CMM error with respect to Gaussian non-common error is of practical interest. In two special cases, analytic approximations to this expectation valid for large number of vehicles can be established.
\section{Asymptotic Analysis for the expected square error}
In this section, the expected square error is derived with two different road configurations. The road configuration is specified by the driving directions of all the roads along which the participating vehicles travel.
\subsection{Orthogonal road directions}
In the first case, it is assumed that each road is parallel to one of the two orthogonal axes of the global reference frame. As a result, the direction angles $\theta$ of the roads relative to the reference frame belong to a set with four elements. Without loss of generality, we define the reference frame to be aligned with the road directions, such that the angles of driving direction
\begin{equation}
\theta_i\in \{0,\frac{\pi}{2}, \pi,\frac{3\pi}{2}\}
\end{equation}
\indent This case can be viewed as a simplified model for the urban areas where most roads are orthogonal.\\
\indent Invoking (8), the square error can be expressed analytically as
\begin{equation}
e^2=\frac{X^2_1+X^2_2+X^2_3+X^2_4-2X_1X_3-2X_2X_4}{4},
\end{equation}
where $X_j,j=1,2,3,4$ are the largest projections of the composite non-common error on each of the four normal vector:
\begin{equation}
X_j=max\{\tilde{x}^N_{j_1}\cdot n_{j},\tilde{x}^N_{j_2}\cdot n_{j},...,\tilde{x}^N_{j_{Nj}}\cdot n_{j}\}, j=1,2,3,4.
\end{equation}
$N_j, j=1,2,3,4$ are the numbers of vehicles traveling in each of the four directions.\\ 
\indent If all the composite non-common error, $\tilde{x}^N_i, i=1,2,...N$, is independent and identically distributed, then according to the Fisher-Tippett-Gnedenko theorem \cite{basrak2011fisher}, the limit distribution of $X_j$ for large $N_j$ is Gumbel distribution whose cumulative distribution function is given by
\begin{equation}
F(X_j)=exp(-exp(-(X_j-\mu_j)/\beta_j)).
\end{equation}
\indent Moreover, with the third assumption, The leading order of the normalization constants $\mu_j$ and $\beta_j$ are related to the variance of the Gaussian distribution $\sigma$ through \cite{david2003order}
\begin{equation}
\mu_j\sim\sigma\sqrt{2log(N_j)},\beta_j\sim\sigma\frac{1}{\sqrt{2log(N_j)}}.
\end{equation}
\indent Using the property of Gumbel distribution, the expectation of $e^2$ with respect to $X_j$ can be evaluated:
\begin{equation}
\begin{aligned}
E_X[e^2]&=\frac{\pi^2}{24}\sum_{j=1}^4\beta^2_j+\frac{1}{4}[\mu_1-\mu_3+\gamma(\beta_1-\beta_3)]^2\\ 
&+\frac{1}{4}[\mu_2-\mu_4+\gamma(\beta_2-\beta_4)]^2,
\end{aligned}
\end{equation}
where $\gamma\approx 0.5772$ is the Euler-Mascheroni constant.\\
\indent Considering the case that all the four $N_j$ are large number of the same order and using the asymptotic formulas (15), it can be readily shown that the first term in (16) is of $O(\frac{1}{log(N_j)})$ and the following two terms are of $O(\frac{1}{N_jlog(N_j)})$. Thus, the leading order asymptotic approximation is
\begin{equation}
E_X[e^2]\sim\frac{\pi^2\sigma^2}{48}\sum_{j=1}^4\frac{1}{log(N_j)}.
\end{equation}

\subsection{Uniformly distributed random road directions}
\indent In the second case, each direction angle is assumed to be randomly distributed within $[0,2\pi)$.
Due to the fourth assumption, Eq. (8) can be linearized with respect to the non-common error:
\begin{equation}
e=e_0+\Delta e=e_0+\frac{C\tilde{X}}{S_0},
\end{equation}
where 
\begin{equation}
e_0=\frac{1}{S_0}\int\limits_{\Omega_0} \tau' dA,
\end{equation}
\begin{equation}
\Omega_0=\{\tau'|\bigcap\limits_{i=1}^{N}\tilde{g}_i(\tau')<0\},
\end{equation}
\begin{equation}
\tilde{X}=[\tilde{x}^N_1\cdot n_1,\tilde{x}^N_2\cdot n_2,...,\tilde{x}^N_N\cdot n_N]^T,
\end{equation}
and 
\begin{equation}
C=S_0\frac{\partial e}{\partial \tilde{X}}.
\end{equation}
$C$ is a $2\times N$ matrix whose components are related to the geometric quantities of the road constraints.\\
\indent The condition under which the linearization (18) is valid is
\begin{equation}
||\tilde{X}||_\infty\ll\frac{2\pi w}{N},
\end{equation}
where $w$ is the half width of the lane.\\
\indent With the assumption that each non-common error obeys independent Gaussian distribution with zero mean, i.e., $\tilde{X}\sim N(0_{N\times 1}, diag(\sigma^2_1,\sigma^2_2,...,\sigma^2_N))$, the expectation of the square error is
\begin{equation}
E_X[e^2]=e^2_0+\frac{1}{S^2_0}tr(L^TC^TCL),
\label{24}
\end{equation}
where $L=diag(\sigma_1,\sigma_2,...,\sigma_N)$ is the Cholesky decomposition of the joint Gaussian covariance matrix.\\
\indent As both $e_0$, $S_0$ and $C$ depends on the road direction angles $\theta_i$, $E_X(e^2)$ is also a random variable.
It can be shown that the expectation of $E_X(e^2)$ with respect to $\theta_i, i=1,2,...,N$ is of $O(\frac{1}{N})$ (See Appendix). More specifically,
\begin{equation}
E_{\theta}[e^2_0]=\frac{2w^2}{9N}+o(\frac{1}{N})
\end{equation}
and
\begin{equation}
E_{\theta}[\frac{1}{S^2_0}tr(L^TC^TCL)]=\frac{3\sum_{i=1}^N \sigma^2_i}{2N^2}+o(\frac{1}{N})
\end{equation}
\section{Simulation Verification}
\indent In this section, simulation results are presented to justify the validity of the asymptotic formulas derived in Section 3. The expectations are calculated by averaging over 5000 samples of $e^2$. Each sample value for the orthogonal road case is calculated through two approaches. One approach is the analytic formula (12), and the other approach is an importance sampling Monte Carlo integration where the proposal distribution is a two-dimensional uniform distribution. Besides, the number of vehicles in each direction is the same. For the uniformly random road angle case, each sample value of $e^2$ is calculated by the Monte Carlo integration. The number of samples to implement each Monte Carlo integration is 10000. The road half width is $w=2\mbox{ }m$ in the simulation.\\
\subsection{Small non-common error: $\sigma=0.3\mbox{ }m$}
\begin{figure}[htbp]
  \centering
  \includegraphics[width=3.5in]{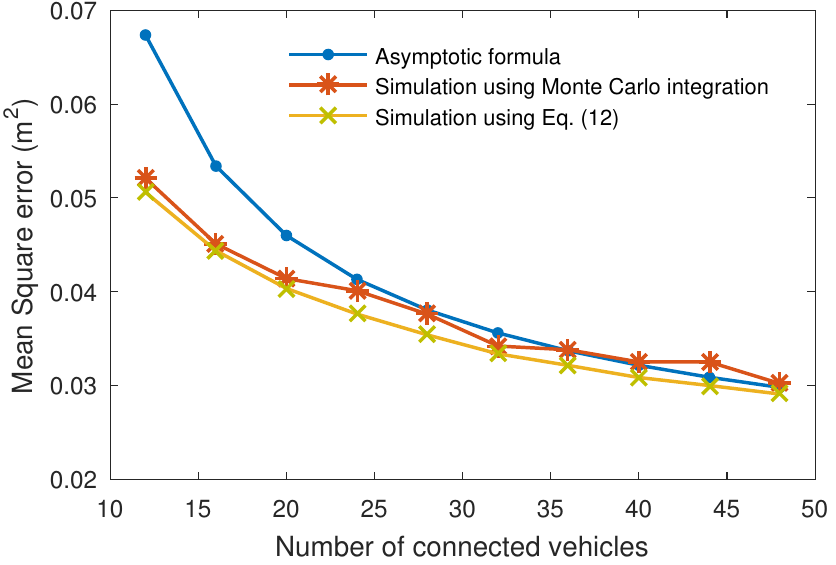}   
  \caption{Comparison between the asymptotic formula and numerical simulation results in the orthogonal road angle case}
  \label{fig. orthogonal}
\end{figure}
\indent Fig. \ref{fig. orthogonal} shows the comparison in the orthogonal road case. The asymptotic formula is in good agreement with the numerical results as the number of vehicles increases. The results using the two numerical approaches are also different, which should be caused by the random error resulted from the Monte Carlo integration used to calculate $e^2$. Therefore, the one that uses (8) is expected to be closer to the underlying true expectation. Compared with this result, the asymptotic formula slightly overestimates the error. This difference may be resulted from the fact that the convergence to the Gumbel distribution is rather slow \cite{david2003order}.

\begin{figure}[htbp]
  \centering
  \includegraphics[width=3.5in]{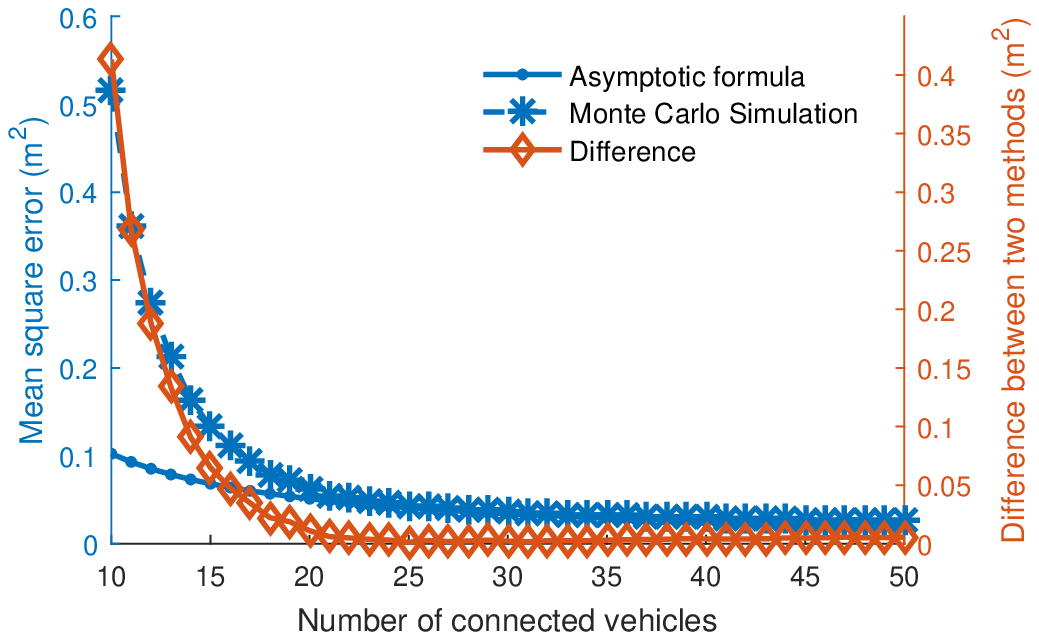}   
  \caption{Comparison between the asymptotic formula and numerical simulation results in the uniformly distributed random road angle case}
  \label{fig_3}
\end{figure}
\indent Fig. \ref{fig_3} shows the comparison in the random road angle case as well as the corresponding difference between the asymptotic formula and the simulation results using Monte Carlo integration. The difference reaches its minimum around $N=25\sim 30$ and increases with the further increase of $N$. This result can be attributed to the following two reasons. First, the asymptotic formula is derived for large $N$. As a result, the difference at small $N$ should be significant. Second, the linearization (18) based on the small non-common error assumption (23) eventually becomes invalid for fixed $\sigma$ and increasing $N$.

\subsection{Large non-common error: $\sigma=1\mbox{ }m$}
\indent The analytic results shown in Section 3 do not apply to the large non-common error case because the approximate feasible set described by (4) may be an empty set. Nonetheless, this problem can be addressed by assigning a weight to each hypothesis of the common error according to its compatibility with the road constraints. The weighted road map approach proposed by Rohani et. al. is applied to generate simulation results because of its simplicity for implementation.\\
\begin{figure}[tb]
  \centering
  \includegraphics[width=3.5in]{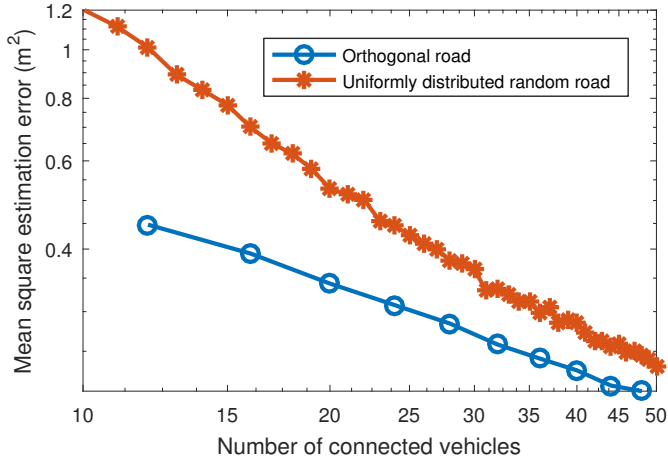}   
  \caption{Numerical results of the error decrement in the two cases of road configuration with large non-common error}
  \label{rf5}
\end{figure}
\indent Fig. \ref{rf5} shows the error decrement in the two different cases of road configuration when the non-common error variance is large. As can be observed from the figure, the slope of the decrement curve in the uniformly distributed road case is steeper than that corresponding to the orthogonal road case. This result is expected from intuition as in the former case, the road angles are more diverse, thereby, providing more constraints to correct the GNSS bias. This can also be understood from the mathematical expression (\ref{eq8}). The CMM error is equal to the deviation of the geometric center enclosed by the road constraints. As the directions of the road angles become diverse, there is a large probability that the error in different directions cancels out. As a result, the expectation of the error become small. In contrast, if there exist some dominant directions, the expectation of the error would be large as the probability that the error cancels out becomes small.
\subsection{Distribution of the mean square CMM error}
\begin{figure}[htbp]
  \centering
  \includegraphics[width=3.5in]{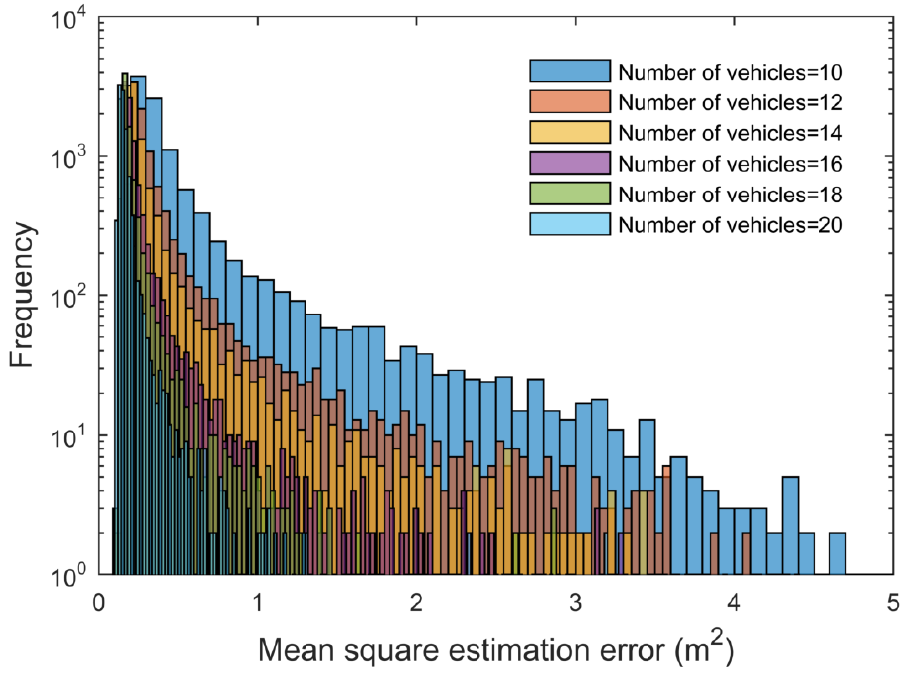}   
  \caption{Histogram of the mean square CMM error, with the road angles drawn from a uniform distribution and the non-common error variance $\sigma^2=0.5 \mbox{ } m^2$ }
  \label{sorted_error}
\end{figure}
The expectation of the mean error is insufficient to describe the full distribution of the error. The distribution of the error is obtained approximately by Monte Carlo simulation. For the number of vehicles ranging from 10 to 20, 10000 samples of road angles are drawn from the uniform distribution and the corresponding mean square CMM error is evaluated and sorted into increasing order shown in Fig. \ref{sorted_error}. As the number of vehicles increases, not only the average error decreases but also the probability of a significant error become small as indicated by the sharp tails of the distribution. 
\section{Optimality of the uniform distribution}
From the asymptotic analysis in Section 3, it is clear that the distribution of the road directions has a significant impact on the mean CMM error.\\
\indent In this section, the following theorem is proved, which quantifies the impact of the road direction distribution on the mean CMM error:

\begin{theorem}
Let the direction angles corresponding to all the connected vehicles be random variables that are independent from each other, obeying the same continuous distribution $p(\theta),\theta\in[0,2\pi)$, then the uniform distribution $p(\theta)=\frac{1}{2\pi}$ is a local minimizer of the mean square CMM error (\ref{24}) in the space of continuous distributions, as the number of connected vehicles $N\xrightarrow{}\infty$.
\end{theorem}
\begin{proof}
\indent Considering an arbitrary continuous distribution of the road angle $p(\theta),\theta\in[0,2\pi)$, the periodic condition should be satisfied,
\begin{equation}
p(0)=p(2\pi^-),
\end{equation}
as $\theta=0$ and $\theta=2\pi$ represent the same angle.\\
\indent This periodicity motivates the following Fourier series expansion,
\begin{equation}
p(\theta)=\frac{1}{2\pi}+\sum_{m\in Z^*}C_mexp(im\theta),
\end{equation}
with
\begin{equation}
C_m=C^*_{-m},
\end{equation}
where the integer $m$ is the summation index, the asterisk denotes the complex conjugate and $i=\sqrt{-1}$ is the imaginary unit. The constant term $\frac{1}{2\pi}$ ensures that the normalization condition is satisfied.\\
\indent In the limiting case that $N\xrightarrow{}\infty$, the leading order term of the square CMM error $e^2_0$ due to deviation of the geometric center can be approximated as
\begin{equation}
e^2_0=\frac{4w^2}{9}\frac{\sum_{i=1}^N \tan^2(\frac{\tilde{\theta_i}}{2})}{\pi^2} \approx \frac{w^2}{9}\frac{\sum_{i=1}^N \tilde{\theta_i}^2}{\pi^2},
\end{equation}
where $\tilde{\theta_i}$ is the difference between two adjacent angles $\theta_{i+1}$ and $\theta_i$ as defined in the appendix.\\
\indent In order to derive the expectation of $e^2_0$, the distribution of $\tilde{\theta_i}$ denoted as $f(\tilde{\theta_i};N,p(\theta_i))$ will be considered first. $f(\tilde{\theta_i};N,p(\theta_i))$ is a nearest neighbor distribution, which satisfies the following integral equation \cite{NN_distribution},
\begin{equation}
f(\tilde{\theta_i})=2Np(\theta_i)(1-\int^{\tilde{\theta_i}}_0f(\tau)d\tau),
\end{equation}
where the dependence on the parameters $N$ and $p(\theta_i)$ will be omitted hereafter.\\
\indent Together with the normalization condition, the solution to (31) is
\begin{equation}
f(\tilde{\theta_i})=2Np(\theta_i)exp(-2Np(\theta_i)\tilde{\theta_i}).
\end{equation}
\indent The number of vehicles $N$ and the local density of the road angles $p(\theta_i)$ appear as parameters in the distribution. As the product $Np(\theta_i)$ increases, the angles distributed around $\theta_i$ become dense, thus increasing the probability of small differential angle $\tilde{\theta_i}$. The expectation of $\tilde{\theta_i}^2$ is
\begin{equation}
E[\tilde{\theta_i}^2]=\int^\infty_0 \tilde{\theta_i}^2 f(\tilde{\theta_i}) d\tilde{\theta_i}=\frac{1}{2N^2p^2(\theta_i)}.
\end{equation}
\indent Combining (30) and (33), the expectation of $e^2_0$ can be derived as follows,
\begin{equation}
E[e^2_0]=\frac{w^2}{9\pi^2}\sum ^N_{i=1} E[\tilde{\theta_i}^2]=\frac{w^2}{36N\pi^2}\sum^N_{i=1}\frac{1}{p^2(\theta_i)}\frac{2\pi}{N}.
\end{equation}
\indent The summation in (34) can be interpreted as an integration as the number of the vehicles $N$ goes to infinity,
\begin{equation}
\begin{aligned}
\lim\limits_{N \to \infty }{\sum^N_{i=1}\frac{1}{p^2(\theta_i)}\frac{2\pi}{N}}&=\lim\limits_{N \to \infty }{\sum^N_{i=1}\frac{1}{p^2(\theta_i)}\Delta \theta}\\
&=\int^{2\pi}_{0}\frac{1}{p^2(\theta)}d\theta.
\end{aligned}
\end{equation}
\indent The Fourier expansion of $\frac{1}{p^2(\theta)}$ can be obtained and expressed in terms of the Fourier coefficients of $p(\theta)$, assuming the deviation from the uniform distribution is infinitesimal,
\begin{equation}
\begin{aligned}
\frac{1}{p^2(\theta)}&=4\pi^2-8\pi^3\sum_{m\in Z^*}C_mexp(im\theta)\\
&+16\pi^4(\sum_{m\in Z^*}C_mexp(im\theta))^2+O(C^3_m).
\end{aligned}
\end{equation}
\indent Substituting (36) into (35),
\begin{equation}
\lim\limits_{N \to \infty }{\sum^N_{i=1}\frac{1}{p^2(\theta_i)}\frac{2\pi}{N}}=8\pi^3+32\pi^5 \sum^{\infty}_{m=1}|C_m|^2.
\end{equation}
\indent Substituting (37) into (34),
\begin{equation}
E[e^2_0]\sim \frac{w^2}{36N\pi^2}(8\pi^3+32\pi^5 \sum^{\infty}_{m=1}|C_m|^2).
\end{equation}
\indent By taking $C_m=0, m\in Z^+$, which corresponds to $p(\theta)=\frac{1}{2\pi}$ (uniform distribution), the expectation of the square error $E[e^2_0]$ is minimized and (25) is also recovered.
\end{proof}

\section{Evaluation of CMM accuracy with safety pilot data}
\begin{figure}[htbp]
\includegraphics[width=\columnwidth]{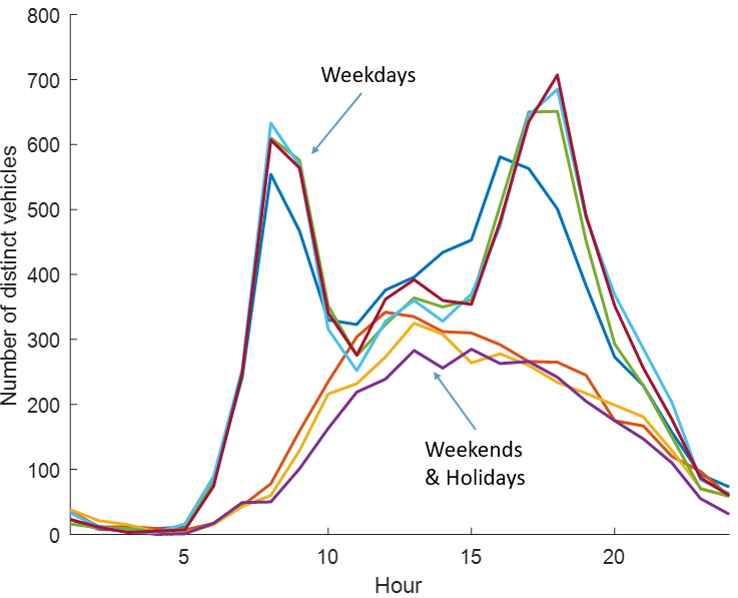}
\caption{Number of recorded vehicles at different hours on different days within a week}
\label{fig. numb_vehicle}
\end{figure}
\begin{figure*}[htbp]
\centering
\includegraphics[width=0.9\textwidth]{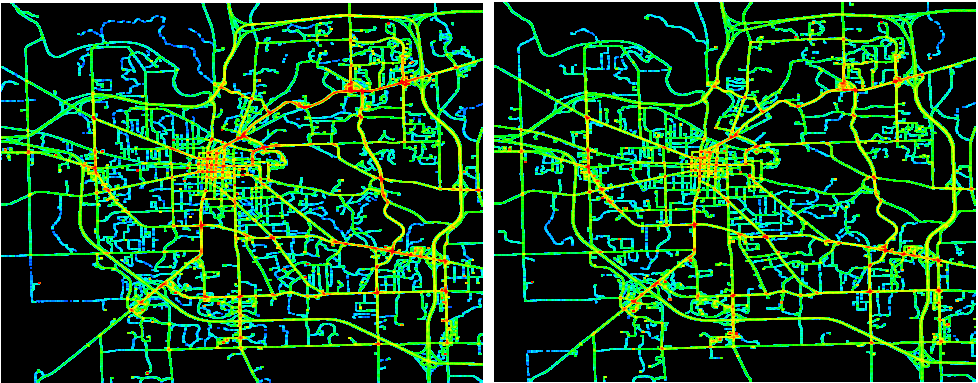}
\caption{Vehicle density distribution on two successive weekdays; Red: dense; Green: medium; Blue: sparse}
\label{fig. density}
\end{figure*}
In this section, the localization accuracy of CMM, using the Safety Pilot program data that were collected in Ann Arbor over three years time period, will be evaluated based on (24).\\ 
\indent The average square error $E_{\theta}[E_X[e^2_0]]$, which is a measurement of localization accuracy, depends on the distribution of the road angles and the number of connected vehicles. These two components vary spatially due to the limited communication radius. It would be sufficient for the evaluation of the localization accuracy at any point if the locations and the road angles of all the vehicles were given. As a result, the root mean square error $J \triangleq \sqrt{E_{\theta}[E_X[e^2_0]]}$ is a functional of the joint distribution of all the vehicle locations and the road angles as well as a function of the location on the map,
\begin{equation}
J(x)=J(x;p(\mathbf{x_v,\Theta})),
\end{equation}
where $x$ shows the location where the localization accuracy is evaluated and $\mathbf{x_v}=[x_1,x_2,...x_N]^T$ denotes the locations of all the vehicles in Ann Arbor and $\mathbf{\Theta}=[\theta_1,\theta_2,...,\theta_N]^T$ denotes the corresponding road angles. It is worth noting that $J(x)$ varies with time as both the joint distribution $p(\mathbf{x_v,\Theta}))$ and the total number of vehicles $N$ varies with time. For example, $N$ will be large at daytime and small at midnight. In order to evaluate $J(x)$ at some given time, e.g., from 8 a.m. to 9 a.m., it is required to estimate the distribution $p(\mathbf{x_v,\Theta}))$ and the number of vehicles $N$ from historical data. The SPMDP1 dataset \cite{SP}, which records a 3-year history of the locations and driving angles of 3000 vehicles in Ann Arbor, is used to estimate the distribution as well as the total number of vehicles, assuming that the recorded vehicles comprise 3\% of the total vehicle population in Ann Arbor. \\
\indent Examples of the temporal distribution of the recorded vehicles are shown in Fig. \ref{fig. numb_vehicle}. On weekdays, there are two peaks at 8 a.m. and 5 p.m., respectively. On weekends and holidays, there are only one peak appearing at 2 p.m.. At midnight, the number of recorded vehicles significantly drops down. This sparsity of vehicles might lead to a significantly biased estimation of the vehicle spatial distribution. In order to reduce this bias, the spatial distribution is obtained by averaging over every day's distribution at the same hour or by averaging over the 24 hours within a day. \\
\indent The spatial distributions of vehicles over two successive weekdays are shown in Fig. \ref{fig. density}. These two distributions looks similar, which indicates the bias due to vehicle sparsity has been mitigated through the averaging.\\
\indent After an unbiased estimation of the vehicle locations and the road angles distribution has been obtained, samples of $(\mathbf{x_v,\Theta})$ are drawn, and $J(x)$ is approximated as the average performance over the sample realizations.

\begin{figure*}[htbp]
\includegraphics[width=\textwidth]{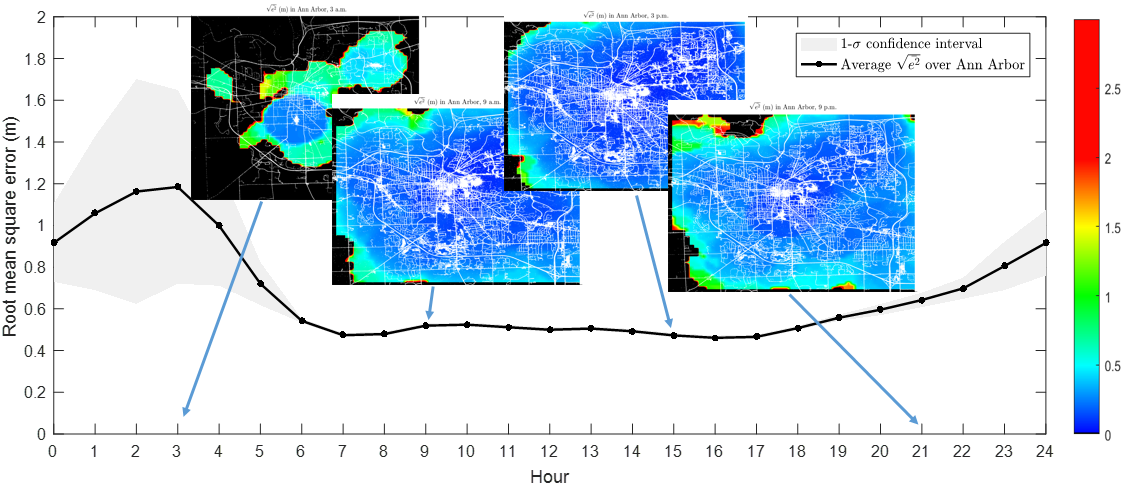}
\caption{Average root mean square CMM error over Ann Arbor at different hours of a day.}
\label{fig. hour}
\includegraphics[width=\textwidth]{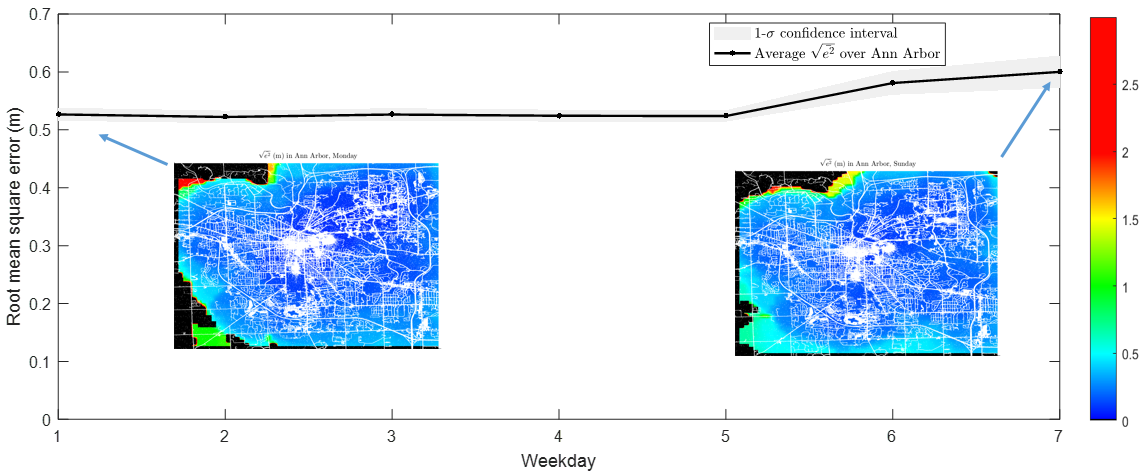}
\caption{Average root mean square CMM error over Ann Arbor at different days of a week}
\label{fig:day}
\includegraphics[width=\textwidth]{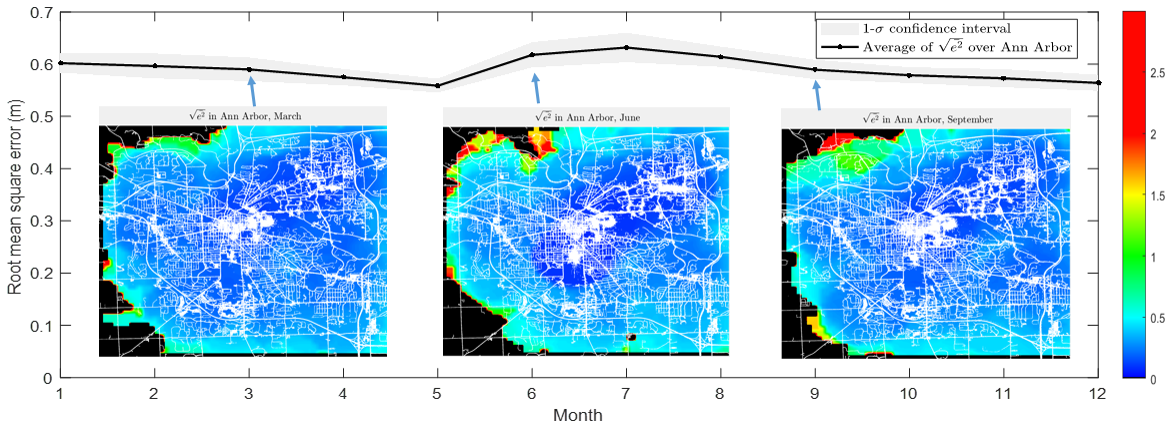}
\caption{Average root mean square CMM error over Ann Arbor at different months of a year}
\label{fig:month}
\end{figure*}
\indent The average root mean square CMM error with its $1-\sigma$ bound at the different time within a day is plotted in Fig. \ref{fig. hour}. The communication radius is assumed to be 1 mile and the non-common error variance is assumed to be a constant $0.5\mbox{ }m^2$. As expected, the CMM error is large at midnight due to the lack of connected vehicles. The variance of the error is also large which means that the performance of CMM is unstable if the number of connected vehicles is inadequate. In contrast, both the performance and the stability of the performance are much better during the daytime when there is a sufficient number of vehicles connected. Fig. \ref{fig. hour} also shows the spatial distribution of the CMM error, with blue indicating a small error, red indicating a large error and black indicating inadequate vehicles for robust localization. In the central area where the vehicles are concentrated and the road angles are diverse, the CMM error is much smaller compared to the marginal area.\\ 
\indent Similarly, the average root mean square CMM error on different days of a week is shown in Fig. \ref{fig:day}, and that on different months of a year is shown in Fig. \ref{fig:month}.\\
\indent The means and variances plotted in Fig. \ref{fig:day}-\ref{fig:month} are listed in the Appendix.
%
%
\section{CONCLUSIONS}
In this paper, the impact of road configuration on the CMM localization accuracy is studied theoretically, which is then used to evaluate the CMM accuracy in real world. The key findings and contributions in this paper are summarized as follow:
\begin{enumerate}
 \item A closed form expression of the CMM error in terms of the road angles and GNSS error is derived based on a simple CMM rule, based on which the asymptotic decay of the CMM error is established analytically. 
 \item The uniform distribution of the road angles is the optimal distribution that minimizes the CMM error on average.
 \item The CMM error in Ann Arbor is predicted based on actual traffic flow data collected in the SPMDP1 dataset.
 \end{enumerate}



\section*{APPENDIX}
\begin{table*}[htbp] 
\caption{Root mean square CMM error (m)} 
\centering      
\begin{tabular}{c c c c c c c c c c c c c c c c c c c c c c c c c}  
\hline\hline                        
Hour &1 &2 &3 &4 &5 &6 &7 &8 &9 &10 &11 &12 \\ [0.1ex]
\hline
Mean (m) &0.916 &1.058 &1.162 &1.185 	&0.999 &0.720 &0.543 &0.474 &0.479 &0.520 &0.524 &0.512\\
Std. (m) &0.187 &0.369 &0.539 &0.464 &0.288 &0.103 &0.010 &0.010 &0.010 &0.010 &0.010 &0.010\\
\hline \hline
Hour &13 &14 &15 &16 &17 &18 &19 &20 &21 &22 &23 &24\\
\hline
Mean (m) &0.500 &0.506 &0.493 &0.473 &0.461 &0.466 &0.508 &0.558 &0.596 &0.642 &0.698 &0.805\\
Std. (m) &0.010 &0.010 &0.010 &0.010 &0.010 &0.010 &0.010 &0.015	&0.027 &0.034 &0.051 &0.116\\
\hline \hline
Month &Jan. &Feb. &Mar. &Apr. &May &Jun. &Jul. & Aug. & Sep. & Oct. & Nov. &Dec.\\ [0.1ex] 
\hline                    
Mean (m) &0.602 &0.596 &0.590 &0.575 &0.559 	&0.618	&0.632 &0.614 &	0.589 &0.579 &0.573 &0.564\\   
Std. (m) &0.018&	0.024 &0.023 &0.016 &0.013 & 0.024 &0.028 &0.019 &0.019 &0.015 &0.017 &	0.015\\ [0.1ex]
\hline \hline    
Weekday &Mon. &Tue. &Wed. &Thu. &Fri. &Sat. &Sun.\\ [0.1ex] 
\hline 
Mean (m) &0.527 &0.522 &0.527 &0.524 &0.524 &0.581 &0.600\\
Std. (m) &0.011 &0.011 &0.011 &0.010 &0.010 &0.021 &0.028\\
\hline \hline
\end{tabular}
\end{table*}
\indent The detail with respect to the derivation of (25) and (26) is presented here.\\
\indent The direction angles of the roads are denoted by $[\theta_1,\theta_2,...,\theta_N]$. Without loss of generality, it is assumed that $0\leq \theta_1 \leq \theta_2 \leq...\leq \theta_N \leq 2\pi$. The increments of the angles are defined as
\begin{equation}
\tilde{\theta_i}=\begin{cases}
\theta_{i+1}-\theta_i & \text{for }i=1,2,...N-1\\
\theta_1-\theta_N+2\pi & \text{for }i=N\\
\end{cases}
\end{equation}
$e^2_0$ can be expressed in terms of the geometric quantities:
\begin{equation}
\begin{aligned}
e^2_0&=\frac{(\sum_{i=1}^N \frac{2}{3}w^3tan(\frac{\tilde{\theta_i}}{2})cos(\theta_i))^2}{S^2_0}\\
&+\frac{(\sum_{i=1}^N \frac{2}{3}w^3tan(\frac{\tilde{\theta_i}}{2})sin(\theta_i))^2}{S^2_0}\\
&=\frac{4w^2}{9}\frac{\sum_{i=1}^N \tan^2(\frac{\tilde{\theta_i}}{2})}{\pi^2}\\
&+\frac{4w^2}{9}\frac{\sum_{i=1}^N\sum_{j=1}^N tan(\frac{\tilde{\theta_i}}{2})tan(\frac{\tilde{\theta_j}}{2})cos(\theta_j-\theta_i)}{\pi^2}\\
&+Higher\mbox{ } Order\mbox{ } Terms
\end{aligned}
\end{equation}
\indent The probability distribution of $\tilde{\theta_i}$ can be derived:
\begin{equation}
p(\tilde{\theta}_i)=\frac{N}{\pi}(1-\frac{\tilde{\theta_i}}{\pi})^{N-1}, \mbox{ }0\leq \tilde{\theta_i}\leq \pi.
\end{equation}
\indent Accurate to the leading order:
\begin{equation}
E_{\theta}[e^2_0]=\frac{4w^2}{9\pi^2}E_{\theta}[\sum_{i=1}^N tan^2(\frac{\tilde{\theta_i}}{2})]=\frac{4w^2N}{9\pi^2}E_{\theta}[tan^2(\frac{\tilde{\theta_i}}{2})]
\end{equation}
\indent The expectation in (30) can be calculated: 
\begin{dmath}
E_{\theta}[tan^2(\frac{\tilde{\theta_i}}{2})]=\int^{\pi}_0tan^2(\frac{\tilde{\theta_i}}{2})\frac{N}\pi(1-\frac{\tilde{\theta_i}}{\pi})^{N-1}d\tilde{\theta_i}\\
=\int^{\frac{\pi}{\sqrt{N}}}_0(\frac{\tilde{\theta_i}}{2})^2\frac{N}\pi(1-\frac{\tilde{\theta_i}}{\pi})^{N-1}d\tilde{\theta_i}\\
+\int^{\frac{\pi}{\sqrt{N}}}_0[tan^2(\frac{\tilde{\theta_i}}{2})-(\frac{\tilde{\theta_i}}{2})^2]\frac{N}\pi(1-\frac{\tilde{\theta_i}}{\pi})^{N-1}d\tilde{\theta_i}\\
+\int^{\frac{\pi}{2}}_{\frac{\pi}{\sqrt{N}}}tan^2(\frac{\tilde{\theta_i}}{2})\frac{N}\pi(1-\frac{\tilde{\theta_i}}{\pi})^{N-1}d\tilde{\theta_i}\\
+\int^{\pi}_{\frac{\pi}{2}}tan^2(\frac{\tilde{\theta_i}}{2})\frac{N}\pi(1-\frac{\tilde{\theta_i}}{\pi})^{N-1}d\tilde{\theta_i}.
\end{dmath}
\indent The first term after the last equality in (31) can be integrated analytically and shown that
\begin{equation}
\int^{\frac{\pi}{\sqrt{N}}}_0(\frac{\tilde{\theta_i}}{2})^2\frac{N}\pi(1-\frac{\tilde{\theta_i}}{\pi})^{N-1}d\tilde{\theta_i}\sim \frac{\pi^2}{2N^2}.
\end{equation}\\
\indent The remaining three term will be shown as $o(\frac{1}{N^2})$ terms:
\begin{equation}
\begin{aligned}
&|\int^{\frac{\pi}{\sqrt{N}}}_0[tan^2(\frac{\tilde{\theta_i}}{2})-(\frac{\tilde{\theta_i}}{2})^2]\frac{N}\pi(1-\frac{\tilde{\theta_i}}{\pi})^{N-1}d\tilde{\theta_i}|\\
&<\int^{\frac{\pi}{\sqrt{N}}}_0(\frac{\tilde{\theta_i}}{2})^4\frac{N}\pi(1-\frac{\tilde{\theta_i}}{\pi})^{N-1}d\tilde{\theta_i}\\
&<\frac{\pi^2}{4N}\int^{\frac{\pi}{\sqrt{N}}}_0(\frac{\tilde{\theta_i}}{2})^2\frac{N}\pi(1-\frac{\tilde{\theta_i}}{\pi})^{N-1}d\tilde{\theta_i}\\
&=O(\frac{1}{N^3}).
\end{aligned}
\end{equation}
\begin{equation}
\begin{aligned}
&|\int^{\frac{\pi}{2}}_{\frac{\pi}{\sqrt{N}}}tan^2(\frac{\tilde{\theta_i}}{2})\frac{N}\pi(1-\frac{\tilde{\theta_i}}{\pi})^{N-1}d\tilde{\theta_i}|\\
&<\int^{\frac{\pi}{2}}_{\frac{\pi}{\sqrt{N}}}\frac{N}\pi(1-\frac{\tilde{\theta_i}}{\pi})^{N-1}d\tilde{\theta_i}\\
&=(1-\frac{1}{\sqrt{N}})^N-(\frac{1}{2})^N.
\end{aligned}
\end{equation}
\indent As
\begin{equation}
\begin{aligned}
\lim_{N\to\infty} N^2(1-\frac{1}{\sqrt{N}})^N&=\lim_{N\to\infty} N^4((1-\frac{1}{N})^N)^N\\
&=\lim_{N\to\infty} N^4e^{-N}=0,
\end{aligned}
\end{equation}
it follows that
\begin{equation}
\int^{\frac{\pi}{2}}_{\frac{\pi}{\sqrt{N}}}tan^2(\frac{\tilde{\theta_i}}{2})\frac{N}\pi(1-\frac{\tilde{\theta_i}}{\pi})^{N-1}d\tilde{\theta_i}=o(\frac{1}{N^2}).\\
\end{equation}
\begin{equation}
\begin{aligned}
&|\int^{\pi}_{\frac{\pi}{2}}tan^2(\frac{\tilde{\theta_i}}{2})\frac{N}\pi(1-\frac{\tilde{\theta_i}}{\pi})^{N-1}d\tilde{\theta_i}|\\
&<\int^{\pi}_{\frac{\pi}{2}}\frac{N}\pi(1-\frac{\tilde{\theta_i}}{\pi})^{N-3}d\tilde{\theta_i}\\
&<\int^{\pi}_{\frac{\pi}{2}}\frac{N}\pi(\frac{1}{2})^{N-3}d\tilde{\theta_i}=\frac{N}{2}(\frac{1}{2})^{N-3}=o(\frac{1}{N^2}).\\
\end{aligned}
\end{equation}
\begin{equation}
\begin{aligned}
&E_{\theta}[\frac{1}{S^2_0}tr(L^TC^TCL)]=E_{\tilde{\theta}}[\frac{1}{S^2_0}tr(L^TC^TCL)]\\
&=\sum_{i=1}^N\frac{2\sigma^2_i(E_{\tilde{\theta}}[tan^2(\frac{\tilde{\theta}_i}{2})]+E_{\tilde{\theta}}^2[tan(\frac{\tilde{\theta}_i}{2}])}{\pi^2}+H.O.T.\\
&=\frac{3\sum_{i=1}^N \sigma^2_i}{2N^2}+o(\frac{1}{N}).\\
\end{aligned}
\end{equation}

\bibliography{ITS}
\bibliographystyle{IEEEtran}

\begin{IEEEbiography}[{\includegraphics[width=1in,height=1.25in,clip,keepaspectratio]{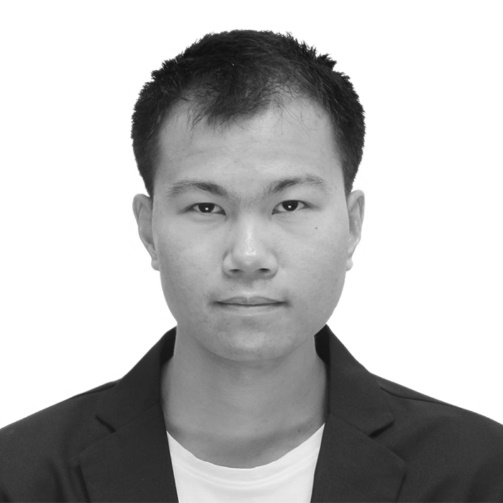}}]{Macheng Shen}
Macheng Shen received his B.S. degree in 2015 from Shanghai Jiao Tong University and his M.S.E. degree in 2016 from the University of Michigan, Ann Arobr. He is currently a Ph.D. student with the department of Naval Architecture and Marine Engineering in the University of Michigan, Ann Arbor. His research interests include connected vehicle localization, Bayesian inference and machine learning.   
\end{IEEEbiography}
\begin{IEEEbiography}[{\includegraphics[width=1in,height=1.25in,clip,keepaspectratio]{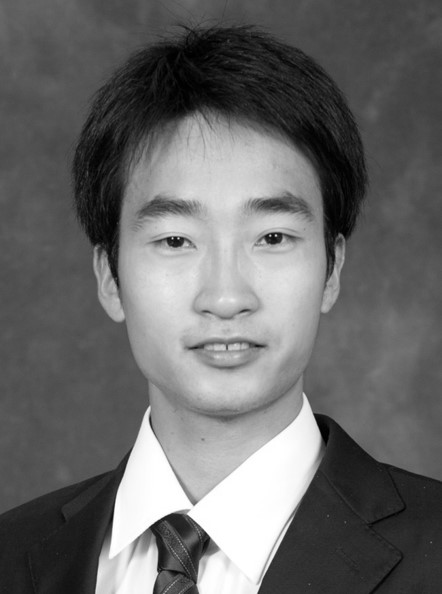}}]{Ding Zhao}
Ding Zhao received the Ph.D. degree in 2016 from the University of Michigan, Ann Arbor. He is currently an Assistant Research Scientist in the Department of Mechanical Engineering at the University of Michigan. His research interest includes evaluation of connected and automated vehicles, vehicle dynamic control, driver behaviors modeling, and big data analysis
\end{IEEEbiography}
\begin{IEEEbiography}[{\includegraphics[width=1in,height=1.25in,clip,keepaspectratio]{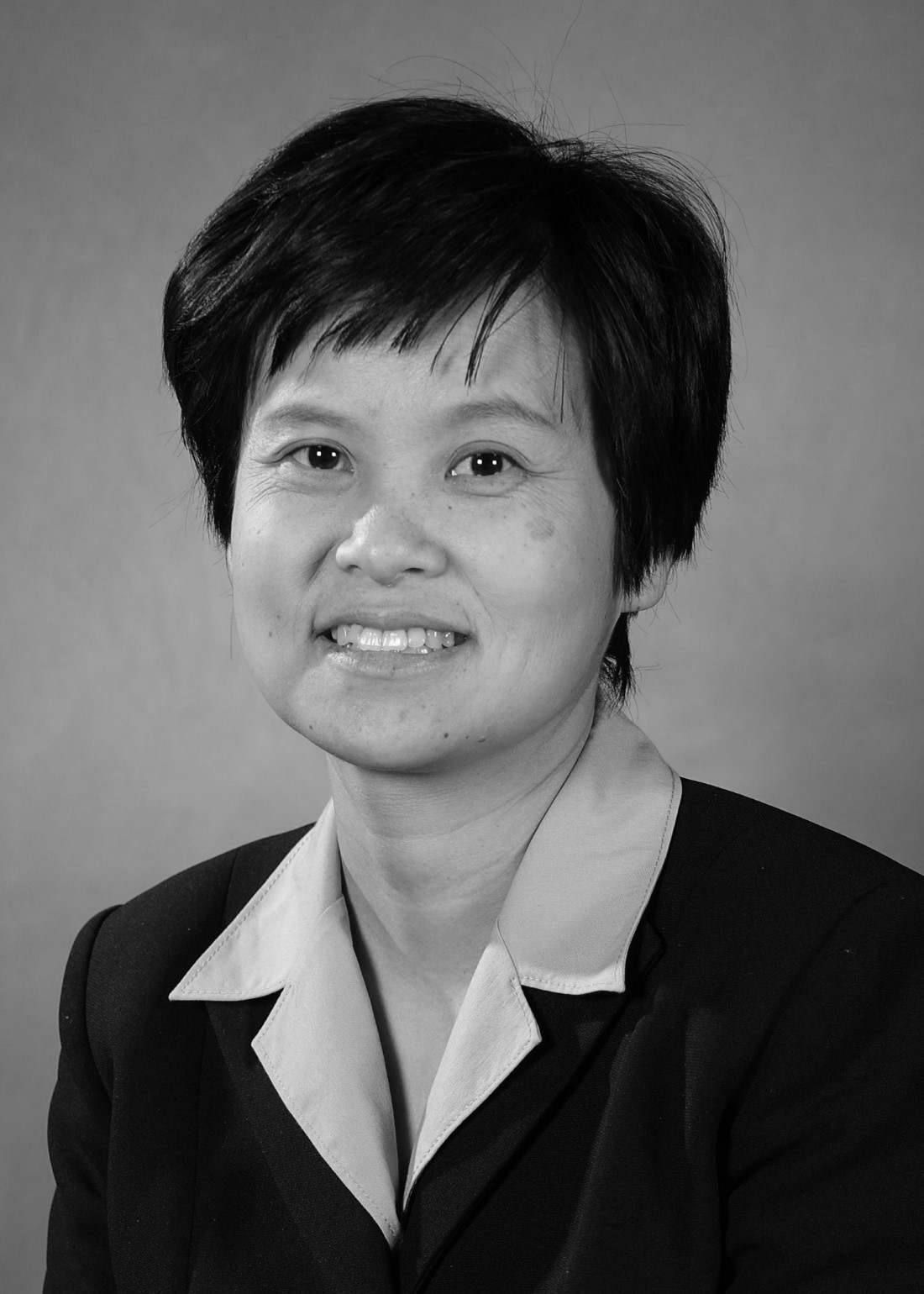}}]{Jing Sun}
Jing Sun received her Ph. D. degree from University of Southern California in 1989, and her B. S. and M. S. degrees from University of Science and Technology of China in 1982 and 1984 respectively. From 1989-1993, she was an assistant professor in Electrical and Computer Engineering Department, Wayne State University. She joined Ford Research Laboratory in 1993 where she worked in the Powertrain Control Systems Department. After spending almost 10 years in industry, she came back to academia and joined the faculty of the College of Engineering at the University of Michigan in 2003, where she is now Micheal G. Parsons Professor in the Department of Naval Architecture and Marine Engineering, with courtesy appointments in the Department of Electrical Engineering and Computer Science and Department of Mechanical Engineering. Her research interests include system and control theory and its applications to marine and automotive propulsion systems. She holds 39 US patents and has co-authored a textbook on Robust Adaptive Control. She is an IEEE Fellow and a recipient of the 2003 IEEE Control System Technology Award.
\end{IEEEbiography}
\end{document}